%% file: root.tex
\documentclass[letterpaper, 10 pt, conference]{ieeeconf}  

\IEEEoverridecommandlockouts                              

\usepackage{url}
\usepackage{cite}
\usepackage{amsmath,amssymb,amsfonts}
\usepackage{algorithmic}
\usepackage{graphicx}
\usepackage{textcomp}
\usepackage[dvipsnames]{xcolor}
\usepackage{soul}
\usepackage{flushend}
\usepackage{multirow}
\usepackage{tikz}
\usepackage{textcomp}
\usepackage{hyperref}

\newcommand\copyrighttext{
  \small \textcopyright 2020 IEEE. Personal use of this material is permitted. Permission from IEEE must be
  obtained for all other uses, in any current or future media, including reprinting/republishing this material for advertising or promotional purposes, creating new collective works, for resale or redistribution to servers or lists, or reuse of any copyrighted component of this work in other works. \newline This article is an extended version of a paper accepted for publication in IEEE Control Systems Letters. Citation information: DOI \href{https://doi.org/10.1109/LCSYS.2020.3003771}{10.1109/LCSYS.2020.3003771}
  }
\newcommand\copyrightnotice{%
\begin{tikzpicture}[remember picture,overlay]
\node[anchor=north,yshift=0.01pt] at (current page.north) {\fbox{\parbox{\dimexpr\textwidth-\fboxsep-\fboxrule\relax}{\copyrighttext}}};
\end{tikzpicture}%
}

\def\BibTeX{{\rm B\kern-.05em{\sc i\kern-.025em b}\kern-.08em
    T\kern-.1667em\lower.7ex\hbox{E}\kern-.125emX}}

\input{macros.tex}

\title{\LARGE \bf Sparse Sensing and Optimal Precision: An Integrated Framework for $\Htwo/\Hinf$ Optimal Observer Design}

\author{Vedang M. Deshpande$^{1,3}$ and Raktim Bhattacharya$^{2,3}$
\thanks{This work was supported by the National Science Foundation (grant
number: 1762825).}
\thanks{$^{1}$Vedang M. Deshpande is a Ph.D. student in Aerospace Engineering. {\tt\small vedang.deshpande@tamu.edu}}%
\thanks{$^{2}$Raktim Bhattacharya is Associate Professor in Aerospace Engineering,
Electrical \& Computer Engineering. {\tt\small raktim@tamu.edu}}
\thanks{$^{3}$Texas A\&M University, College Station, TX 77843, USA.}%
}

\begin{document}

\maketitle
\thispagestyle{empty}
\copyrightnotice
\pagestyle{empty}

\begin{abstract}
In this paper, we simultaneously determine the optimal sensor precision and the observer gain, which achieves the specified accuracy in the state estimates. Along with the unknown observer gain, the formulation parameterizes the scaling of the exogenous inputs that correspond to the sensor noise. Reciprocal of this scaling is defined as the sensor precision, and sparseness is achieved by minimizing the $l_1$ norm of the precision vector. The optimization is performed with constraints guaranteeing specified accuracy in state estimates, which are defined in terms of $\Htwo$ or $\Hinf$ norms of the error dynamics. The results presented in this paper are applied to the linearized longitudinal model of an F-16 aircraft.
\end{abstract}

\section{INTRODUCTION}
The conventional observer design deals with the problem of determining observer gain for a system, given the set of sensors with pre-specified precision, to achieve the desired performance index. Here the precision is related to the sensor noise signal, and can be quantified by the inverse of variance or $\mathcal{L}_2$-norm of the signal.
Often in control system design the sensors are pre-selected and the performance of control and estimation algorithms are limited by this choice. Therefore, it may be possible that unnecessarily precise sensors are included in the system, for a required performance. Or if more performance is desired, it is unclear which sensors to improve, or even where to add new sensors. For large-scale systems, this question becomes difficult and non trivial. We address this problem in the context of state-estimation for LTI (linear time invariant) systems.

In this paper, we consider the problem of selecting a sparse set of sensors and simultaneously determining the minimum required precision, for observer design for LTI systems. The problem is formulated in $\Htwo/\Hinf$ optimal estimation framework, and posed as a convex optimization problem. This problem is not new and considerable amount of work exists in the literature \cite{boyd2009tac, Nugroho2018acc, Jovanovic2014cdc, Jovanovic2018cdc, Jovanovic2019TAC, Roy2013CAMSAP,  Polyak2013ecc, Lopez2014acc, Wolfrum2014tac, Sundaram2017automatica, das2017icssa, das2020ifac, skelton2008jour, saraf2017acc, Matni2016tac,Matni2019tac}.

In \cite{boyd2009tac}, authors formulated the sensor selection problem as a Boolean convex optimization problem and relaxed it by allowing parameters to vary continuously between 0 and 1. A parameter is set to zero if it comes out to be less than a pre-specified value while maximizing the confidence ellipsoid of the unbiased estimate. A framework for simultaneous sensor and actuator selection while ensuring stability in terms of Boolean variables was proposed in \cite{Nugroho2018acc}.

The formulations presented in  \cite{Jovanovic2014cdc, Jovanovic2018cdc, Polyak2013ecc,Jovanovic2019TAC,Wolfrum2014tac} augment the cost function with sparsity-promoting penalty on the columns of observer gain matrix (rows of controller gain matrix) to get a sparse set of sensors (actuators). 
Works in \cite{das2017icssa,das2020ifac,Sundaram2017automatica,Lopez2014acc,Wolfrum2014tac} considered minimal sensor selection for discrete time systems. 
A discussion on system level approach to control/sensing architecture design with sparsity constraints can be found in \cite{Matni2016tac,Matni2019tac}.

Aforementioned papers \cite{Nugroho2018acc, boyd2009tac, Jovanovic2014cdc, Jovanovic2018cdc, Jovanovic2019TAC, Sundaram2017automatica, Lopez2014acc, Roy2013CAMSAP, Polyak2013ecc, Wolfrum2014tac} assume that the precision of sensors is known or fixed. On the other hand, the framework proposed in \cite{skelton2008jour} treats the sensor and actuator precision as design variables to be determined, while guaranteeing the optimal controller performance. The work in \cite{skelton2008jour} was extended for models with parametric uncertainty in \cite{saraf2017acc}.

\subsubsection*{\text{Contribution and novelty}}
\responseFour{The primary focus of this paper is to present an integrated theoretical framework to design $\Htwo/\Hinf$ optimal observers with sparse sensor configurations, while simultaneously minimizing the required sensor precision.
Motivated by \cite{skelton2008jour}, in this paper, we treat sensor precision as an unknown variable, unlike existing sparse sensor selection frameworks discussed above.
We consider the  $\Htwo/\Hinf$ optimal observer design problem for continuous LTI systems with a specified  performance criterion.
The objective here is twofold. First, we are interested in minimizing the sensor precision, and second, we want to obtain a sparse sensor configuration. The optimal precision for sensors is determined by minimizing the sparsity-promoting $l_1$-norm of the precision vector.
The aforementioned frameworks obtain sparse sensor configuration by inducing column-sparseness in the observer gain, assuming that the sensor precisions are given. In our work we induce sparseness by directly scaling the individual sensor channels, and simultaneously determine the observer gain for those precisions. To the best of our knowledge, this is the first integrated formulation for designing $\Htwo/\Hinf$ optimal observers.}


The paper is organized as follows. The sparse $\Htwo/\Hinf$ observer design problems are formulated in \S \ref{sec:prob}. Solutions to the observer design problems are presented in \S \ref{sec:thms} as Theorems \ref{thm:h2} and \ref{thm:hinf}. In \S \ref{sec:ex}, we consider a numerical example, followed by the concluding remarks in \S \ref{sec:concl}.

\section{Problem Formulation} \label{sec:prob}
\subsection{Notation}
The set of real numbers is denoted by $\Real$. Matrices (vectors) are denoted by bold uppercase (lowercase) letters e.g. $\A, \Bu$ ($\x$, $\y$). $\A^T$ denotes the transpose of $\A$. We define $\textbf{sym}\left(\A\right):=\A+\A^T$. We use the notation $\A>0$ ($\A<0$) for symmetric positive (negative) definite matrices. For an integer $N>0$, $\I{N}$ denotes the $N\times N$ identity matrix. Zero matrix of suitable dimensions is denoted by $\vo{0}$. For any $r\in \Real$, $\x^{r}$ denotes the vector with element-wise powers raised to $r$. A diagonal matrix constructed from $\x$ is denoted by $\diag(\x)$. Similarly, $\diag\left(\A_1,\A_2,\cdots,\A_N \right)$ denotes the block diagonal matrix. 
\subsection{System and observer}
Consider the following LTI system
\begin{subequations}
\begin{align}
    \xdot(t) &= \A\x(t) + \Bu\u(t) + \Bw\w(t), \\
    \y(t) &= \Cy\x(t) + \Du\u(t) + \Dw\w(t), \\
    \z(t) &= \Cz\x(t),
\end{align} \eqnlabel{sys}
\end{subequations}
where, $\x\in\Real^{\nx}$, $\y\in\Real^{\ny}$, $\z\in\Real^{\nz}$ are respectively the state vector, the vector of measured outputs, and the output vector of interest.  The vector of control inputs is denoted as $\u\in\Real^{\nuu}$, and  $\w\in\Real^{\nw}$ is the vector of disturbance signals partitioned as
\begin{align}
    \w(t) :=\begin{bmatrix} \vo{d}(t) \\ \n(t) \end{bmatrix}, \nonumber
\end{align}
where, $\vo{d} \in \Real^{\nd}$ is the process noise, and $\vo{n} \in \Real^{\nv}$ is the sensor noise. The real matrices $\A, \Bu, \Bw, \Cy, \Cz, \Du$, and $\Dw$ are system matrices of appropriate dimensions.

Let us consider the full-order state observer for the system \eqn{sys} given by
\begin{subequations}
\begin{align}
    \dot{\hat{\x}}(t) =& \left(\A+\Lg\Cy \right)\hat{\x}(t) - \Lg\y(t) \nonumber \\ & +\left(\Bu+\Lg\Du\right)\u(t), \\
    \hat{\z}(t) =& \Cz\hat{\x}(t),
\end{align}
\eqnlabel{obs}
\end{subequations}
where, $\hat{\x}\in\Real^{\nx}$ is the estimate of the state vector, $\hat{\z}\in\Real^{\nz}$ is the estimate of the output vector of interest, and the $\Lg \in \Real^{\nx \times \ny}$ is the observer gain.
Let us define the error vectors as
\begin{align*}
    \xerr(t) := \x(t)-\hat{\x}(t), \text{ and }
    \zerr(t) := \z(t)-\hat{\z}(t).
\end{align*}
Therefore, from equations \eqn{sys} and \eqn{obs}, the observation error system can be written as
\begin{subequations}
\begin{align}
    \dot{\e}(t) &= \left(\A+\Lg\Cy\right)\xerr(t) + \left(\Bw+\Lg\Dw\right)\w(t),\\
    \zerr(t) &= \Cz\xerr(t) .
\end{align} \eqnlabel{obs_err}
\end{subequations}
The objective is to determine the gain matrix $\Lg \in \Real^{\nx \times \ny}$ such that $\left(\A+\Lg\Cy \right)$ is stable, and the effect of $\w$ on $\zerr$ is minimal.

The matrices $\Bw$ and $\Dw$ in \eqn{obs_err} can be partitioned as
\begin{align}
    \Bw = \begin{bmatrix}\Bd & \Bv\end{bmatrix}, \text{ and }
    \Dw = \begin{bmatrix}\Dd & \Dv\end{bmatrix}. \eqnlabel{BD_part}
\end{align}
The process is assumed to be independent of sensor noise, i.e. $\Bv=\vo{0}$, and individual sensor channels are independent of each other, i.e. $\Dv=\I{\nv}$.
Now, we define the scaled disturbance signal $\bar{\w}(t)$  as
\begin{align}
    \bar{\w}(t) :=\begin{bmatrix} \bar{\vo{d}}(t) \\ \bar{\n}(t) \end{bmatrix}, \text{  such that, }
  \w(t) = \underbrace{\begin{bmatrix} \Sd & \vo{0}\\ \vo{0} & \Sv \end{bmatrix}}_{=:\Sw} \bar{\w}(t), \eqnlabel{scale}
\end{align}
where, $\Sd \in\Real^{\nd\times\nd}$, $\Sv \in\Real^{\nv\times\nv}$ are constant diagonal scaling matrices with non-negative elements.

The plant ($P$) and estimator ($E$) system  given by equations \eqn{sys} and \eqn{obs} with scaled disturbance is shown in \fig{plant_est}.
\begin{figure}[htb]
    \centering
    \includegraphics[width=0.32\textwidth]{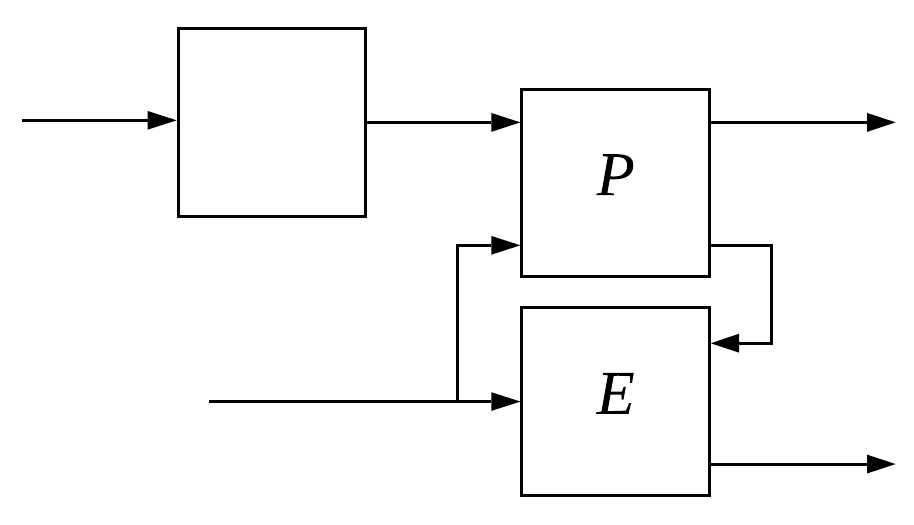}
    \begin{picture}(0,0)
        \put(-20,15){$\hat{\z}(t)$}
        \put(-27,40){$\y(t)$}
        \put(-20,75){$\z(t)$}
        \put(-97,75){$\w(t)$}
        \put(-165,75){$\bar{\w}(t)$}
        \put(-135,28){$\u(t)$}
        \put(-122,68){$\Sw$}
    \end{picture}
    \caption{The plant and estimator system.}
    \label{fig:plant_est}
\end{figure}

Let $\vo{\kappa} =\left[\kappa_1,\cdots, \kappa_{\nv}\right]^T \in \Real^{\nv}$ such that
\begin{align}
   & \diag(\vo{\kappa}) := \Sv^{-1} \implies (\Sv\Sv^T)^{-1} = \diag(\vo{\kappa}^2).  \eqnlabel{alp2_def} 
\end{align}
Combining equations \eqn{obs_err}, \eqn{BD_part} and \eqn{scale}  yields
\begin{subequations}
\begin{align}
    \dot{\e}(t) = & \left(\A+\Lg\Cy\right)\e (t)\nonumber \\ &+ \Big(\underbrace{\begin{bmatrix}\Bd\Sd & \vo{0}\end{bmatrix}}_{=:\Bwb} +\Lg \underbrace{\begin{bmatrix}\Dd\Sd & \Sv\end{bmatrix}}_{=:\Dwb}\Big)\bar{\w}(t), \eqnlabel{Bwb}\\
    \zerr(t) = & \Cz\xerr(t) .
\end{align} \eqnlabel{scaled_obs_err}
\end{subequations}
The transfer function of the system \eqn{scaled_obs_err} is given by
\begin{align*}
    \Gwz := & \Cz\big(s\I{\nx}-\A-\Lg\Cy\big)^{-1} \big(\Bwb +\Lg\Dwb\big),
\end{align*}
where, $s$ is the complex variable.

\subsection{Sensor precision}
\response{In this work, we model $\bar{\vo{d}}(t)$ and $\bar{\n}(t)$ as either zero-mean stationary stochastic processes, or norm bounded signals.}

\response{First, let us consider the case when $\bar{\vo{d}}(t)$ and $\bar{\n}(t)$ are power signals modeled as zero-mean stationary stochastic processes. Let us define auto-correlation matrix of $\bar{\n}(t)$ as $\bar{\vo{\Sigma}}_n(\tau):=\E{\bar{\n}(t+\tau)\bar{\n}^T(t)}$, where $\E{\cdot}$ denotes the expectation operator. Since individual sensor channels are independent of each other,  $$\bar{\vo{\Sigma}}_n(\tau):=\diag(\bar{\sigma}^2_1(\tau), \bar{\sigma}^2_2(\tau), \cdots, \bar{\sigma}^2_{\ny}(\tau))$$ is a diagonal matrix.
Using \eqn{alp2_def}, the auto-correlation matrix of $\n(t) = \Sv\bar{\n}(t)$ becomes
$$\vo{\Sigma}_n(\tau):= \Sv\bar{\vo{\Sigma}}_n(\tau)\Sv^T = \diag\left(\frac{\bar{\sigma}^2_1(\tau)}{\kappa_1^2}, \cdots, \frac{\bar{\sigma}^2_{\ny}(\tau)}{\kappa_{\ny}^2}\right).$$
The signal variance or power of $i^{\text{th}}$ sensor noise channel is given by $\bar{\sigma}^2_i(0)/\kappa_i^2$. Therefore, the precision of $i^{\text{th}}$ sensor channel, which is defined to be the inverse of the signal variance, becomes $\kappa_i^2/\bar{\sigma}^2_i(0) = \kappa_i^2$, where, without loss of generality, we assume $\bar{\sigma}^2_i(0)=1$. Therefore, $\vo{\kappa}^2$ is the precision vector.}

\response{Now, let us consider another case when $\bar{\vo{d}}(t)$ and $\bar{\n}(t)$ are norm bounded but arbitrary signals. Let $\bar{\n}_i(t)$ denote the $i^{\text{th}}$ component of the noise vector $\bar{\n}(t)$. Therefore, using  \eqn{alp2_def}, $\mathcal{L}_2$-norm of the $i^{\text{th}}$ component of the noise vector $\n(t) = \Sv\bar{\n}(t)$ becomes $\norm{\n_i(t)}{2} =  \norm{\bar{\n}_i(t)}{2}/\kappa_i$. In this case, we define sensor precision as the inverse of square of $\mathcal{L}_2$-norm of the noise signal, i.e. the precision of $i^{\text{th}}$ sensor channel is defined as $1/\norm{\n_i(t)}{2}^2 = \kappa_i^2/\norm{\bar{\n}_i(t)}{2}^2 = \kappa_i^2$, where, without loss of generality, we assume $\norm{\bar{\n}_i(t)}{2}=1$. Again, $\vo{\kappa}^2$ is the precision vector.} 

\subsection{Observer design problem}
\response{A sparse sensor configuration can be obtained by making $\vo{\kappa}^2$ sparse, since a sensor with zero precision is equivalent to removing that sensor from the system. Ideally, minimizing $\norm{\vo{\kappa^2}}{0}$, i.e. number of non-zero elements in $\vo{\kappa}^2$, will yield the sparsest sensor configuration. Minimization of $\norm{\cdot}{0}$ is a non-convex problem and the computational cost can be very high. Moreover, $\norm{\cdot}{0}$ does not penalize the magnitude of elements of $\vo{\kappa}^2$, therefore, as shown in \S \ref{sec:ex}, the sensor precision to realize the sparsest configuration can be prohibitively large.}

\response{Generally, higher sensor precision implies higher economic cost. Therefore, economic cost can be characterized by $l_1$-norm of $\vo{\kappa}^2$, i.e. $\norm{\vo{\kappa^2}}{1}$. Moreover, minimization of $l_1$-norm promotes sparsity in the sensor configuration. Therefore, in the proposed approach where sensor precision $\vo{\kappa}^2$ is treated as an unknown variable, we minimize $\norm{\vo{\kappa^2}}{1}$,  or in a general setting, weighted $l_1$-norm of $\vo{\kappa}^2$. For any arbitrary vector $\vo{\beta}\in\Real^N$, its weighted $l_1$-norm is defined as $$\norm{\vo{\beta}}{1,\vo{\rho}} := \sum_{i=1}^{N} \rho_i|\beta_i|$$
where $\vo{\rho}:=[\rho_1,\rho_2,\cdots,\rho_{N}]^T\geq 0$ are pre-defined weights.}

We are interested in determining the sparse set of sensors and associated minimum precision to design the observer given by \eqn{obs}, such that the effect of $\bar{\w}$ on $\zerr$ is minimal.
Therefore, the $\Htwo$ optimal observer design problem for a given attenuation level $\gamma >0$, is then stated as:
\begin{align}
    \text{determine optimal $\vo{\kappa}$ and $\Lg$ s.t. } \norm{\Gwz}{2} < \gamma. \eqnlabel{h2_prob}
\end{align}
Similarly, the $\Hinf$ optimal observer design problem for a given $\gamma >0$ is:
\begin{align}
     \text{determine optimal $\vo{\kappa}$ and $\Lg$ s.t. } \norm{\Gwz}{\infty} < \gamma. \eqnlabel{hinf_prob}
\end{align}
In \eqn{h2_prob} and \eqn{hinf_prob}, `optimal $\vo{\kappa}$ and $\Lg$' minimize $\norm{\vo{\kappa^2}}{1,\vo{\rho}}$, which serves a dual purpose as discussed above. Next, we formally present the solution of observer design problems as theorems.

\section{Sparse $\Htwo/\Hinf$ observers} \label{sec:thms}
\subsection{Main result}
The following theorem solves the $\Htwo$ optimal observer design problem with sparse sensing.

\begin{theorem} \label{thm:h2}
The solution of sparse $\Htwo$ observer design problem \eqn{h2_prob} is determined by solving the following optimization problem. The solution is given by $\vo{\kappa} = \vo{\beta}^{1/2}$, and $\Lg = \X^{-1}\Y$.

\begin{equation}\left.
\begin{aligned}
  & \min\limits_{\Y,\Q>0,\X>0,\vo{\beta}>0}\quad  \norm{\vo{\beta}}{1,\vo{\rho}}  \quad \text{ such that }  \\ 
  & \M_{11}=\big(\X\A+\Y\Cy\big)+\big(\X\A+\Y\Cy\big)^T, \\
 &\M_{12}=\X\Bd\Sd  +\Y\Dd\Sd, \\
&\begin{bmatrix} \M_{11} & \M_{12} & \Y \\
                    \M_{12}^T & -\I{\nd} & \vo{0} \\
                   \Y^T & \vo{0} & -\diag(\vo{\beta}) \end{bmatrix} < 0,\\
 &\begin{bmatrix} -\Q & \Cz \\ \Cz^T & -\X \end{bmatrix} < 0, \\
 & \trace{\Q}  <\gamma^2 .
\end{aligned}\right\}\eqnlabel{h2_thm}
\end{equation}
\end{theorem}
\begin{proof}
The condition $\norm{\Gwz}{2} < \gamma$ in \eqn{h2_prob} is equivalent to the existence of a symmetric matrix $\P>0$ such that \cite{lmiCSys}
\begin{subequations}
\begin{align}
   &\big(\A+\Lg\Cy\big)\P  + \P \big(\A+\Lg\Cy\big)^T \nonumber \\ & \quad \quad  + \big(\Bwb +\Lg\Dwb\big) \big(\Bwb +\Lg\Dwb\big)^T < 0, \eqnlabel{ineq_ap} \\
   &\trace{\Cz\P\Cz^T}  < \gamma^2 . \eqnlabel{ineq_trace}
\end{align}
\end{subequations}
Pre- and post-multiplying \eqn{ineq_ap} by $\P^{-1}$gives
\begin{align}
   &\P^{-1}\big(\A+\Lg\Cy\big)  + \big(\A+\Lg\Cy\big)^T\P^{-1} \nonumber \\ & \quad \quad  + \P^{-1}\big(\Bwb +\Lg\Dwb\big) \big(\Bwb +\Lg\Dwb\big)^T\P^{-1} < 0. \nonumber
\end{align}
Let us substitute $\X:=\P^{-1}$, and $\Y:=\X\Lg$ in the previous equation to get
\begin{align}
    & \big(\X\A+\Y\Cy\big)+\big(\X\A+\Y\Cy\big)^T + \nonumber \\
    & \quad \quad  + \big(\X\Bwb +\Y\Dwb\big) \big(\X\Bwb +\Y\Dwb\big)^T < 0. \eqnlabel{h2_lmi_1}
\end{align}
 Using the definitions of $\Bwb$ and $\Dwb$ from \eqn{Bwb}, and defining
 $\M_{11}:=\big(\X\A+\Y\Cy\big)+\big(\X\A+\Y\Cy\big)^T$, and $\M_{12}:=\X\Bd\Sd  +\Y\Dd\Sd$, inequality \eqn{h2_lmi_1} can be written as
\begin{align}
     &\M_{11} + \begin{bmatrix}\M_{12} &  \Y\Sv\end{bmatrix}  \begin{bmatrix}\M_{12}^T\\  \Sv^T\Y^T\end{bmatrix}  < 0 \nonumber
\end{align}
\begin{align}
    \text{or, }\ &\M_{11}+\begin{bmatrix}\M_{12}  &  \Y\end{bmatrix}  \begin{bmatrix}\I{\nd} & \vo{0}\\ \vo{0} & \Sv\Sv^T\end{bmatrix} \begin{bmatrix}  \M_{12}^T \\ \Y^T\end{bmatrix} < 0. \nonumber 
\end{align}
Then using Schur complement lemma, equation \eqn{alp2_def}, and the substitution $\vo{\beta}:=\vo{\kappa}^2$, the previous inequality becomes 
\begin{align}
\begin{bmatrix} \M_{11} & \M_{12} & \Y \\
                    \M_{12}^T & -\I{\nd} & \vo{0} \\
                   \Y^T & \vo{0} & -\diag(\vo{\beta}) \end{bmatrix} < 0. \eqnlabel{h2_lmi}
\end{align}

Now consider the inequality \eqn{ineq_trace}, which is equivalent to
\begin{align}
    \Cz\P\Cz^T -\Q < 0, \quad \trace{\Q} < \gamma^2 \nonumber
\end{align}
for a matrix $\Q>0$. Again using Schur complement lemma, and substituting $\P^{-1}=\X$, we get
\begin{align}
    \begin{bmatrix} -\Q & \Cz \\ \Cz^T & -\X \end{bmatrix} < 0,  \quad \trace{\Q} < \gamma^2 . \eqnlabel{h2_trace_lmi}
\end{align}
The set of inequalities given by \eqn{h2_lmi} and \eqn{h2_trace_lmi} define the LMI feasibility conditions for the problem \eqn{h2_prob}.
Therefore, the solution to the problem \eqn{h2_prob} is given by solving the optimization problem
$$\min\limits_{\Y,\Q>0,\X>0,\vo{\beta}>0}\quad  \norm{\vo{\beta}}{1,\vo{\rho}} \text{ subject to \eqn{h2_lmi}, \eqn{h2_trace_lmi} }.$$
\end{proof}

We next present the result for solving the $\Hinf$ optimal observer design problem with sparse sensing.

\begin{theorem} \label{thm:hinf}
The solution of sparse $\Hinf$ observer design problem \eqn{hinf_prob} is determined by solving the following optimization problem. The solution is given by $\vo{\kappa} = \gamma^{-1/2}\vo{\beta}^{1/2}$, and $\Lg = \X^{-1}\Y$.

\begin{equation}\left.
\begin{aligned}
& \min\limits_{\Y,\X>0,\vo{\beta}>0}\quad \norm{\vo{\beta}}{1,\vo{\rho}} \quad  \text{ such that } \\ 
  & \M_{11}=\big(\X\A+\Y\Cy\big)+\big(\X\A+\Y\Cy\big)^T, \\
 &\M_{12}=\X\Bd\Sd  +\Y\Dd\Sd, \\
&\begin{bmatrix} \M_{11} & \M_{12} & \Cz^T & \Y \\
               \M_{12}^T & -\gamma\I{\nd} & \vo{0}    & \vo{0} \\
                \Cz &  \vo{0}    & -\gamma\I{\nz}     & \vo{0} \\
               \Y^T &  \vo{0}    & \vo{0} & -\diag(\vo{\beta}) \end{bmatrix} < 0. \\
\end{aligned} \right\}\eqnlabel{hinf_thm}
\end{equation}
\end{theorem}
\begin{proof}
The condition $\norm{\Gwz}{\infty} < \gamma$ in \eqn{hinf_prob} is equivalent to the existence of a symmetric matrix $\X>0$ such that \cite{Hinf1997paper}
\begin{align}
    \begin{bmatrix} \textbf{sym}\big(\X\A+\X\Lg\Cy\big) +  \Cz^T\Cz & \X\big(\Bwb +\Lg\Dwb\big)  \\
                    \responseThree{\big(\Bwb +\Lg\Dwb\big)^T\X} & -\gamma^2 \I{(\nd+\nv)} \end{bmatrix} < 0. \nonumber
\end{align}
Define $\Y:=\X\Lg$ to get
\begin{align}
    \begin{bmatrix} \textbf{sym}\big(\X\A+\Y\Cy\big) + \Cz^T\Cz & \X\Bwb +\Y\Dwb  \\
                    \big(\X\Bwb +\Y\Dwb\big)^T & -\gamma^2 \I{(\nd+\nv)} \end{bmatrix} < 0 .\nonumber
\end{align}
Using Schur complement lemma, it can be written as
\begin{align}
   & \big(\X\A+\Y\Cy\big) + \big(\X\A+\Y\Cy\big)^T + \Cz^T\Cz + \nonumber \\ &\big(\X\Bwb +\Y\Dwb\big)\gamma^{-2} \I{(\nd+\nv)} \big(\X\Bwb +\Y\Dwb\big)^T  < 0. \eqnlabel{hinf_lmi_1}
\end{align}
Using the definitions of $\Bwb$ and $\Dwb$ from \eqn{Bwb}, and defining
 $\M_{11}:=\textbf{sym}\big(\X\A+\Y\Cy\big),$ and $\M_{12}:=\X\Bd\Sd  +\Y\Dd\Sd,$ the inequality \eqn{hinf_lmi_1} becomes


 \begin{align} &\M_{11} +
 \begin{bmatrix}\M_{12} & \Cz^T & \Y\end{bmatrix} \R \begin{bmatrix} \M_{12}^T \\ \Cz \\ \Y^T\end{bmatrix} < 0 , \nonumber
\end{align}
where, $\R:= \diag\left(\gamma^{-2}\I{\nd}, \ \I{\nz}, \ \gamma^{-2}\Sv\Sv^T \right)$.
Using \eqn{alp2_def}, and Schur complement lemma again, we get

\begin{align} \begin{bmatrix} \M_{11} & \M_{12} & \Cz^T & \Y \\
               \M_{12}^T & -\gamma^2 \I{\nd} & \vo{0}    & \vo{0} \\
                \Cz &  \vo{0}    & -\I{\nz}     & \vo{0} \\
                \Y^T &  \vo{0}    & \vo{0} & -\gamma^2\diag(\vo{\kappa}^2) \end{bmatrix} < 0. \eqnlabel{hinf_lmi1}
\end{align}
Define $\X':=\gamma^{-1}\X$, $\Y':=\X'\Lg$,
 $\M_{11}':=\textbf{sym}\big(\X'\A+\Y'\Cy\big)$, $\M_{12}':=\X'\Bd\Sd  +\Y'\Dd\Sd$, and
$$\F:=\diag\Big(\frac{1}{\sqrt{\gamma}}\I{\nx}, \ \frac{1}{\sqrt{\gamma}}\I{\nd}, \ \sqrt{\gamma}\I{\nz}, \  \frac{1}{\sqrt{\gamma}}\I{\nv} \Big).$$
Then pre- and post-multiply \eqn{hinf_lmi1} by $\F$ and $\F^T$ respectively, and substitute $\vo{\beta}=\gamma\vo{\kappa}^2$ to get
\begin{align} \begin{bmatrix} \M_{11}' & \M_{12}' & \Cz^T & \Y' \\
               \left(\M_{12}'\right)^T & -\gamma \I{\nd} & \vo{0}    & \vo{0} \\
                \Cz &  \vo{0}    & -\gamma\I{\nz}     & \vo{0} \\
                \left(\Y'\right)^T &  \vo{0}    & \vo{0} & -\diag(\vo{\beta}) \end{bmatrix} < 0. \eqnlabel{hinf_lmi}
\end{align}
Clearly, inequality \eqn{hinf_lmi} is equivalent to the one in \eqn{hinf_thm}.
Similar to Theorem \ref{thm:h2}, the solution to sparse $\Hinf$ observer is determined by minimizing the weighted $l_1$-norm $\norm{\vo{\beta}}{1,\vo{\rho}}$ subject to \eqn{hinf_lmi}, which concludes the proof.
\end{proof}

\subsection{Iterative refinement} \label{sec:iterate}
\responseOne{In general, solving semi-definite programs (SDPs) given by \eqn{h2_thm} or \eqn{hinf_thm} does not result in exactly sparse $\vo{\beta}$, i.e. some elements of $\vo{\beta}$ would be relatively small but not exactly zero. However, iterative techniques  with weighted $l_1$-norm minimization \cite{boyd2008weightedL1,Jovanovic2014cdc} can be employed to ensure that the elements of $\vo{\beta}$ are close to zero within specified tolerance. 
To achieve the sparse configuration, \eqn{h2_thm} and \eqn{hinf_thm} are solved multiple times and weights are updated as $\rho^{(k+1)}_i = (\epsilon+\lambda|\beta_i^{(k)}|)^{-1}$, where $\beta_i^{(k)}$is the solution at the end of $k^{\text{th}}$ iteration, a small number $\epsilon >0$ and a constant $\lambda > 0$ are used to ensure that the weights are well-defined at each step. Initial weights are chosen to be unity, i.e. $\rho^{(0)}_i = 1$.}

\responseOne{Similar to \cite{Jovanovic2016EJC,Jovanovic2019TAC}, once we have the sparse structure, the final refined or polished solution is obtained by removing the sensor channels with small precision and re-solving \eqn{h2_thm} or \eqn{hinf_thm} with unit weights, i.e. $\rho_i = 1$.}

\responseTwo{Solutions of SDPs, in general, do not scale well as dimension of the problem is increased. Solution algorithms based on proximal gradient method or ADMM \cite{boyd2011admm} such as presented in \cite{Jovanovic2019TAC, Jovanovic2014cdc, Jovanovic2018cdc} might provide an efficient and scalable alternative for solving such problems. However, development of such customized algorithms is out of the scope of this paper, and will be addressed in our future work}.

\subsection{Normalized system}
The control inputs, exogenous signals, and outputs of a plant are generally multiplied by weighting matrices for normalization. Such system with normalizing weights can be written as
\begin{subequations}
\begin{align}
    \xdot(t) &= \A\x(t) + \underbrace{\Bu\W_u}_{=:\tilde{\B}_u}\tilde{\u}(t) + \underbrace{\Bw\W_w}_{=:\tilde{\B}_w}\tilde{\w}(t), \\
    \y(t) &= \Cy\x(t) + \underbrace{\Du\W_u}_{=:\tilde{\D}_u}\tilde{\u}(t) + \underbrace{\Dw\W_w}_{\tilde{\D}_w}\tilde{\w}(t), \\
    \tilde{\z}(t) &= \underbrace{\W_z\Cz}_{=:\tilde{\C}_z}\x(t),
\end{align} \eqnlabel{sys_norm}
\end{subequations}
where, $\tilde{\u}$, $\tilde{\w}$, $\tilde{\z}$ are normalized vectors, 
and $\W_u$, $\W_w$, and $\W_z$ are the corresponding weighting matrices. It is clear that the results of Theorems \ref{thm:h2} and \ref{thm:hinf} can be used for a system given by \eqn{sys_norm} with augmented system matrices $\tilde{\B}_u$, $\tilde{\B}_w$, $\tilde{\D}_u$, $\tilde{\D}_w$, and $\tilde{\C}_z$. Next, we consider an example to demonstrate the application of results presented in this section.

\subsection{Augmented cost function}
The results presented in Theorems \ref{thm:h2} and \ref{thm:hinf} are derived for a given value of $\gamma$. In practice, we are also concerned with determining the minimum level of attenuation $\gamma$. This can be done easily by augmenting the cost function as
\begin{align}\min \quad \norm{\vo{\beta}}{1,\vo{\rho}} + c \gamma \eqnlabel{aug_cost}\end{align}
where $c\geq0$ is a known weighting constant.
Needless to say, any linear constraints in terms of $\vo{\beta}$, e.g. upper bounds, can be easily incorporated in the optimization problem.

\section{Example}\label{sec:ex}
Let us consider the longitudinal model of an F-16 aircraft \cite{Stevens1992}. The states are velocity $V(ft/s)$, angle of attack $\alpha(rad)$, pitch angle $\theta(rad)$, and pitch rate $q(rad/s)$. The engine thrust force $F(lb)$ and elevator angle $\delta_e(deg)$ are the control inputs.
On board sensors measure body acceleration $\dot{u}\,(ft/s^2)$ along roll axis, body acceleration $\dot{w}\,(ft/s^2)$ along yaw axis, angle of attack $\alpha(rad)$, pitch rate $q(rad/s)$, and dynamic pressure $\bar{q}:=\rho_{atm} V^2/2\, (lb/ft^2)$, where $\rho_{atm}$ is the atmospheric density.
Therefore, state $\x$, control $\u$, and measured output $\y$ vectors are defined as
\begin{align*}\x &:=\begin{bmatrix}V & \alpha & \theta & q\end{bmatrix}^T,\\
\u &:= \begin{bmatrix}F & \delta_e\end{bmatrix}^T, \\
\y &:= \begin{bmatrix}\dot{u} & \dot{w} & \alpha & q & \bar{q} \end{bmatrix}^T.
\end{align*}

The \responseOne{dynamic equations} and outputs are non-linear functions of the states and controls.
The linearized model is obtained at an equilibrium or trim point for steady-level flight condition, with trim velocity $V^\ast=1000 \,ft/s$ at an altitude of $10,000 \,ft$. The states and controls at the trim point are
\begin{align*}
    \x^\ast &=\begin{bmatrix}1000, & -3.02\times10^{-3}, & -3.02\times10^{-3}, & 0\end{bmatrix}^T, \\
    \u^\ast &= \begin{bmatrix}6041.20 & -1.38\end{bmatrix}^T.
\end{align*}
System matrices for the linearized model, and weighting matrices are given in the appendix. We assume that the process noise enters the linearized plant due to fluctuations in the elevator setting. Therefore, $\Bd=\Bu\left[0\quad 1\right]^T$.
We also assume that $\Cz = \I{4}$, and since we are interested in determining scaling for sensors only, we set constant $\Sd = 1$.

Next, we utilize the results from Theorems \ref{thm:h2} and \ref{thm:hinf} to determine the sparse sensor configuration and their precision for the system under consideration represented by \eqn{sys_norm}.  The SDPs \eqn{h2_thm} and \eqn{hinf_thm} are solved using the solver \texttt{SDPT3} \cite{sdpt3} with \texttt{CVX} \cite{cvx} as a parser. \\
\begin{figure}[htb]
    \centering
    \includegraphics[trim=2.6cm 0.7cm 3.5cm 0.2cm,width=0.42\textwidth]{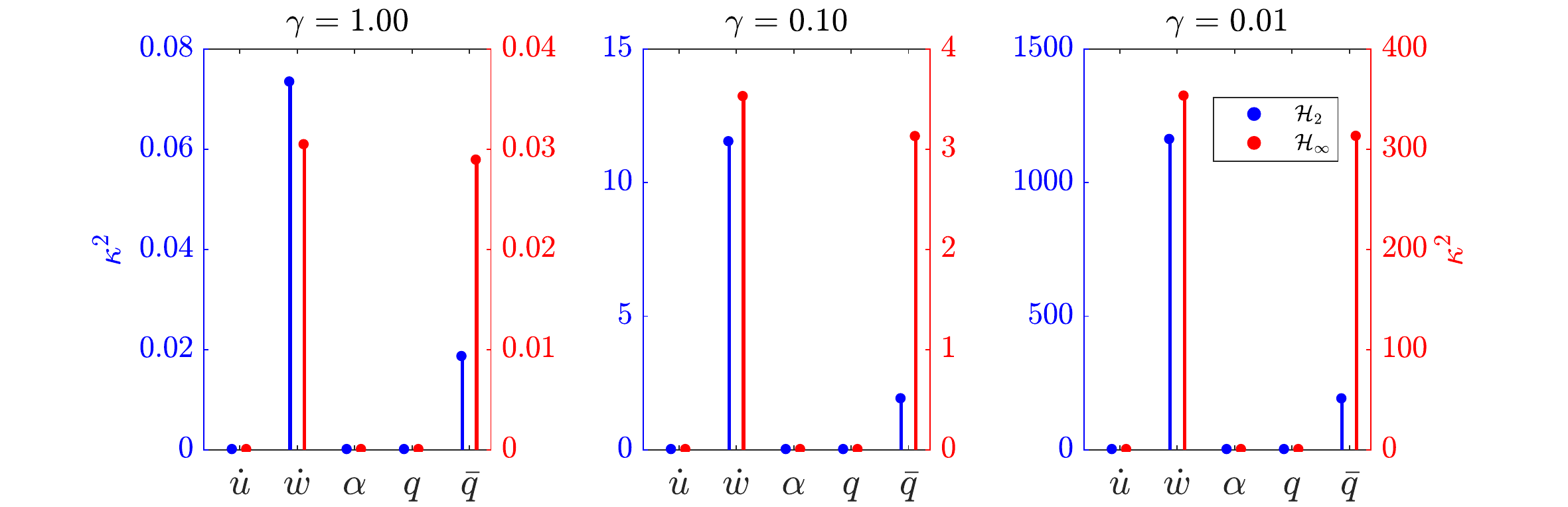}
    \caption{Iterative solution: $\Htwo$ (blue) and $\Hinf$ (red) optimal precision $\vo{\kappa}^2$, $V^{\ast} = 1000 f t/s$, $\norm{\vo{\kappa}^2}{0} = 2$.} \label{fig:polished_h2_hinf}
\end{figure}

\subsubsection*{Iterative solution} The optimization problems are solved iteratively as discussed in \S \ref{sec:iterate} for different values of $\gamma = 1,\,0.1,\,0.01$, with no bounds on $\vo{\kappa}^2$.
The refined sensor precision, $\kappa^2$, of different sensors obtained for $\Htwo$ ($\Hinf$) sparse observer design are shown in \fig{polished_h2_hinf} on left (right) y-axis in blue (red) color.
The precision $\kappa^2$ associated with $\dot{u}$, $\alpha$, and $q$ sensors is zero in all cases. Thus, it implies that, to design the $\Htwo$/$\Hinf$ observer for the plant under consideration, we need $\dot{w}$ and $\bar{q}$ sensors only, and $\norm{\vo{\kappa}^2}{0} = 2$. 
As one would expect, it can be observed from \fig{polished_h2_hinf} that the minimum required precision for sensors increases as the specified value of $\gamma$ is decreased. \\

\subsubsection*{Sensor configurations for different linearized models} The numerical results shown in \fig{polished_h2_hinf} are obtained for the trim velocity $V^{\ast} = 1000 f t/s$.
We performed numerical tests for different trim velocities ranging from $600 f t/s$ to $1600 f t/s$, and observed that the sparse sensor configuration is same in all cases, i.e. $\kappa^2$ is zero for sensors $\dot{u}$, $\alpha$, and $q$.
As one would expect, the values of non-zero $\kappa^2$ for sensors $\dot{w}$ and $\bar{q}$ are different
for different linearized plants. The optimal precision values for different linearized plants are shown
in Tables \ref{table:H2_diff_vel} and \ref{table:Hinf_diff_vel}. We can select sensors with maximum precision that will work
for all linearized plants. However, we also note that the scope of this paper is limited to linear time invariant (LTI) systems. Sparse sensing for non-linear systems is our future research focus.

\responseOne{We compare the iterative solution shown in \fig{polished_h2_hinf} with the solution obtained by exhaustive search, which is discussed next.} \\

\begin{table}[h!]
\centering
\caption{Iterative solution of \eqn{h2_thm}: $\Htwo$ Optimal precision $\kappa^2$ obtained for linearized models at different trim velocities $V^{\ast}\ (ft/s)$.  }
\begin{tabular}{|c|c|c|c|c|c|c|}
 \hline
  $V^{\ast}$& $\gamma$ &  $\kappa^2_1\ (\dot{u})$ & $\kappa^2_2\ (\dot{w})$ & $\kappa^2_3\ (\alpha)$ & $\kappa^2_4\ (q)$ & $\kappa^2_5\ (\bar{q})$ \\
 \hline
  \multirow{3}{1.5em}{600}     & 1    & 0 & 0.0418  & 0 & 0 & 0.0117  \\
                              & 0.1  & 0 & 18.2490 & 0 & 0 & 2.9759  \\
                              & 0.01 & 0 & 1888.3814 & 0 & 0 & 295.7141 \\
 \hline
 \multirow{3}{1.5em}{800}      & 1    & 0 & 0.0686  & 0 & 0 & 0.0170  \\
                              & 0.1  & 0 & 15.1956 & 0 & 0 & 2.6149  \\
                              & 0.01 & 0 & 1539.7173 & 0 & 0 & 261.1855 \\
 \hline
  \multirow{3}{1.5em}{1000}    & 1    & 0 & 0.0733  & 0 & 0 & 0.0186 \\
                              & 0.1  & 0 & 11.5177 & 0 & 0 & 1.9002 \\
                              & 0.01 & 0 & 1160.0183 & 0 & 0 & 189.6549 \\
 \hline
  \multirow{3}{1.5em}{1200}   & 1    & 0 & 0.0739  & 0 & 0 & 0.0169 \\
                              & 0.1  & 0 & 9.8970 & 0 & 0 & 1.6838 \\
                              & 0.01 & 0 & 993.6162 & 0 & 0 & 168.2124\\
 \hline
\multirow{3}{1.5em}{1400}      & 1    & 0 & 0.0723  & 0 & 0 & 0.0154 \\
                              & 0.1  & 0 & 8.7747 & 0 & 0 & 1.5313 \\
                              & 0.01 & 0 & 879.5519 & 0 & 0 & 153.0527\\
 \hline
\multirow{3}{1.65em}{1600}      & 1    & 0 & 0.0699  & 0 & 0 & 0.0143 \\
                              & 0.1  & 0 & 7.9630 & 0 & 0 & 1.4194 \\
                              & 0.01 & 0 & 797.5027 & 0 & 0 & 141.9043\\
 \hline
\end{tabular}
\label{table:H2_diff_vel}
\end{table}

\begin{table}[h!]
\centering
\caption{Iterative solution of \eqn{hinf_thm}: $\Hinf$ Optimal precision $\kappa^2$ obtained for linearized models at different trim velocities $V^{\ast}\ (ft/s)$.}
\begin{tabular}{|c|c|c|c|c|c|c|}
 \hline
  $V^{\ast}$& $\gamma$ &  $\kappa^2_1\ (\dot{u})$ & $\kappa^2_2\ (\dot{w})$ & $\kappa^2_3\ (\alpha)$ & $\kappa^2_4\ (q)$ & $\kappa^2_5\ (\bar{q})$ \\
 \hline
  \multirow{3}{1.5em}{600}     & 1    & 0 & 0.0071 & 0 & 0 & 0.0263 \\
                                & 0.1  & 0 & 6.2214 & 0 & 0 & 5.7935 \\
                                & 0.01 & 0 & 628.1772 & 0 & 0 & 582.4734 \\
 \hline
 \multirow{3}{1.5em}{800}       & 1    & 0 & 0.0305 & 0 & 0 & 0.0375 \\
                                & 0.1  & 0 & 6.6973 & 0 & 0 & 5.3845 \\
                                & 0.01 & 0 & 673.4359 & 0 & 0 & 540.0908 \\
 \hline
 \multirow{3}{1.5em}{1000}      & 1    & 0 & 0.0304 & 0 & 0 & 0.0289 \\
                                & 0.1  & 0 & 3.5216 & 0 & 0 & 3.1227 \\
                                & 0.01 & 0 & 352.6480 & 0 & 0 & 312.5102 \\
 \hline
  \multirow{3}{1.5em}{1200}     & 1    & 0 & 0.0302 & 0 & 0 & 0.0264 \\
                                & 0.1  & 0 & 3.3316 & 0 & 0 & 2.7850 \\
                                & 0.01 & 0 & 333.4735 & 0 & 0 & 278.6482 \\
 \hline
  \multirow{3}{1.5em}{1400}     & 1    & 0 & 0.0303 & 0 & 0 & 0.0246 \\
                                & 0.1  & 0 & 3.2458 & 0 & 0 & 2.5562 \\
                                & 0.01 & 0 & 324.8010 & 0 & 0 & 255.7163 \\
 \hline
  \multirow{3}{1.75em}{1600}     & 1    & 0 & 0.0307 & 0 & 0 & 0.0233 \\
                                & 0.1  & 0 & 3.2287 & 0 & 0 & 2.3957 \\
                                & 0.01 & 0 & 323.0274 & 0 & 0 & 239.6375 \\
 \hline
\end{tabular}
\label{table:Hinf_diff_vel}
\end{table}


\subsubsection*{Exhaustive search}
\responseOne{The globally sparsest $\vo{\kappa}^2$ is determined via exhaustive search as follows. First assume that $\norm{\vo{\kappa}^2}{0} = r$. This gives us $5!/r!(5-r)!$ cases since there are 5 sensors. Problems \eqn{h2_thm} and \eqn{hinf_thm} are solved for each scenario retaining $r$ sensors and removing the other $5-r$, with $\rho_i = 1$. Select sensor configuration with the minimum $r$. If there are multiple feasible solutions for such $r$, then select one with the minimum $\norm{\vo{\kappa}^2}{1}$ as the optimal solution.}

\responseOne{The $\Htwo$ optimal sensor precision obtained via exhaustive search is identical to the iterative solution shown in \fig{polished_h2_hinf}. Therefore, for the system under consideration, the proposed approach produces the globally sparsest configuration with the least $\norm{\vo{\kappa}^2}{1}$ for the $\Htwo$ optimal observer.}
\begin{table}[h!]
\centering
\caption{Exhaustive search: $\Hinf$ optimal precision $\vo{\kappa}^2$, $V^{\ast} = 1000 f t/s$, $\norm{\vo{\kappa}^2}{0} = 1$.}
\begin{tabular}{|c|c|c|c|c|c|c|}
 \hline
  & $\gamma$ &  $\kappa^2_1\ (\dot{u})$ & $\kappa^2_2\ (\dot{w})$ & $\kappa^2_3\ (\alpha)$ & $\kappa^2_4\ (q)$ & $\kappa^2_5\ (\bar{q})$ \\
 \hline
   \multirow{3}{1.5em}{$\Hinf$} & 1    & 0 & 0 & 0 & 0 & 1.2922$\times 10^3$ \\
                                & 0.1  & 0 & 0 & 0 & 0 & 2.5533$\times 10^5$ \\
                                & 0.01 & 0 & 0 & 0 & 0 & 2.5662$\times 10^7$ \\
 \hline
\end{tabular}
\label{table:exhaust}
\end{table}

\responseOne{The $\Hinf$ optimal sensor precision obtained via exhaustive search is shown in Table \ref{table:exhaust}, which is evidently different from the iterative solution shown in \fig{polished_h2_hinf}. Note, the required precision or $\norm{\vo{\kappa}^2}{1}$ to realize  the configuration obtained by iterative solution (\fig{polished_h2_hinf}) is orders of magnitude smaller than the globally sparsest configuration (Table \ref{table:exhaust}).
This is due to the fact that the proposed framework minimizes the individual sensor precision while simultaneously promoting a sparse configuration, whereas the exhaustive search enforces the sparse configuration first and then determines the corresponding sensor precision. This also exposes a trade-off between a sparse configuration and the sensor precision required to realize it.}

\responseOne{Although we considered unbounded $\vo{\kappa}^2$ for the purpose of numerical experiments, in practice, there will be upper bounds on $\vo{\kappa}^2$ arising due to physical constraints. Upper bounds on $\vo{\kappa}^2$ (i.e. linear constraints on $\vo{\beta}$) can be easily incorporated in the optimization problems \eqn{h2_thm} and \eqn{hinf_thm}. Upper bounds on $\vo{\kappa}^2$ can avoid configurations with potentially prohibitive precision such as shown in Table \ref{table:exhaust}.}
\\

\subsubsection*{Simulation of error dynamics}
\responseTwo{To analyze the performance of observers, we simulate the error dynamics given in \eqn{obs_err} with non-zero initial condition. The scaled disturbances $\bar{\w}$ in \eqn{scale} are assumed to be mutually independent unit variance band-limited Gaussian stationary processes. A comparison of $\Htwo$ observers with sparse and full sensor configurations for $\gamma = 0.1$ is shown in \fig{h2_obs_compare}. The sensor precision for full configuration is obtained by solving \eqn{h2_thm} once with $\rho_i=1$, while sensor precision for sparse configuration is shown in \fig{polished_h2_hinf}. From \fig{h2_obs_compare}, we see that the performance of observers with sparse and full configurations are comparable. The reason is that, in full configuration, the sensor precision for $\dot{w}$ and $\bar{q}$ are very close to values  in \fig{polished_h2_hinf}, and the precision for $\dot{u}$, $\alpha$, $q$ sensors is order of magnitude of $10^{-10}$ (not shown here), i.e. sensor precision in sparse and full configurations are very similar.}

\responseTwo{A similar comparison for $\Hinf$ observers is shown in \fig{hinf_obs_compare}, and again, we see that the performance of observers with sparse and full configurations are comparable. }
\begin{figure}[htb]
    \centering
    \includegraphics[trim=1.8cm 0.5cm 3cm 0.4cm,width=0.38\textwidth]{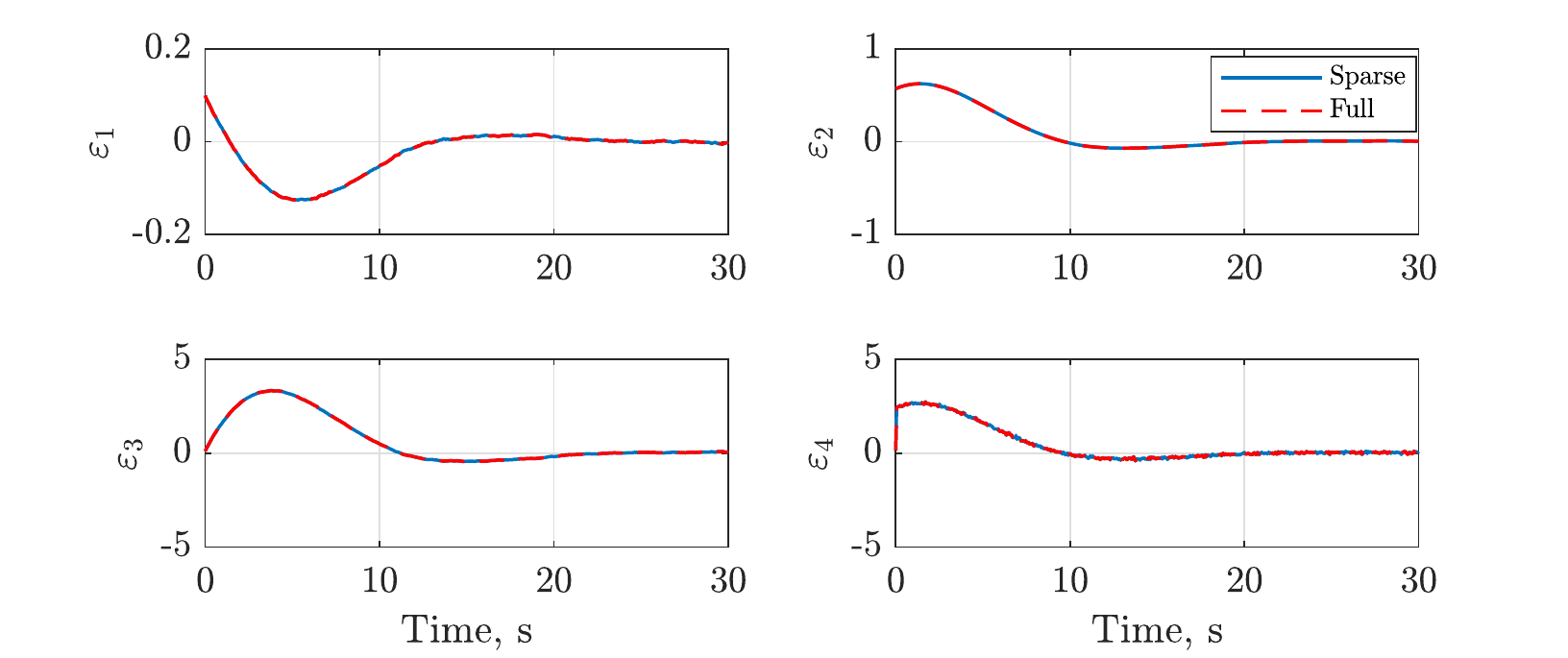}
    \caption{$\Htwo$ observer performance for sparse and full sensor configurations.} \label{fig:h2_obs_compare}
\end{figure}
\begin{figure}[htb]
    \centering
    \includegraphics[trim=1.8cm 0.5cm 3cm 0.6cm,width=0.38\textwidth]{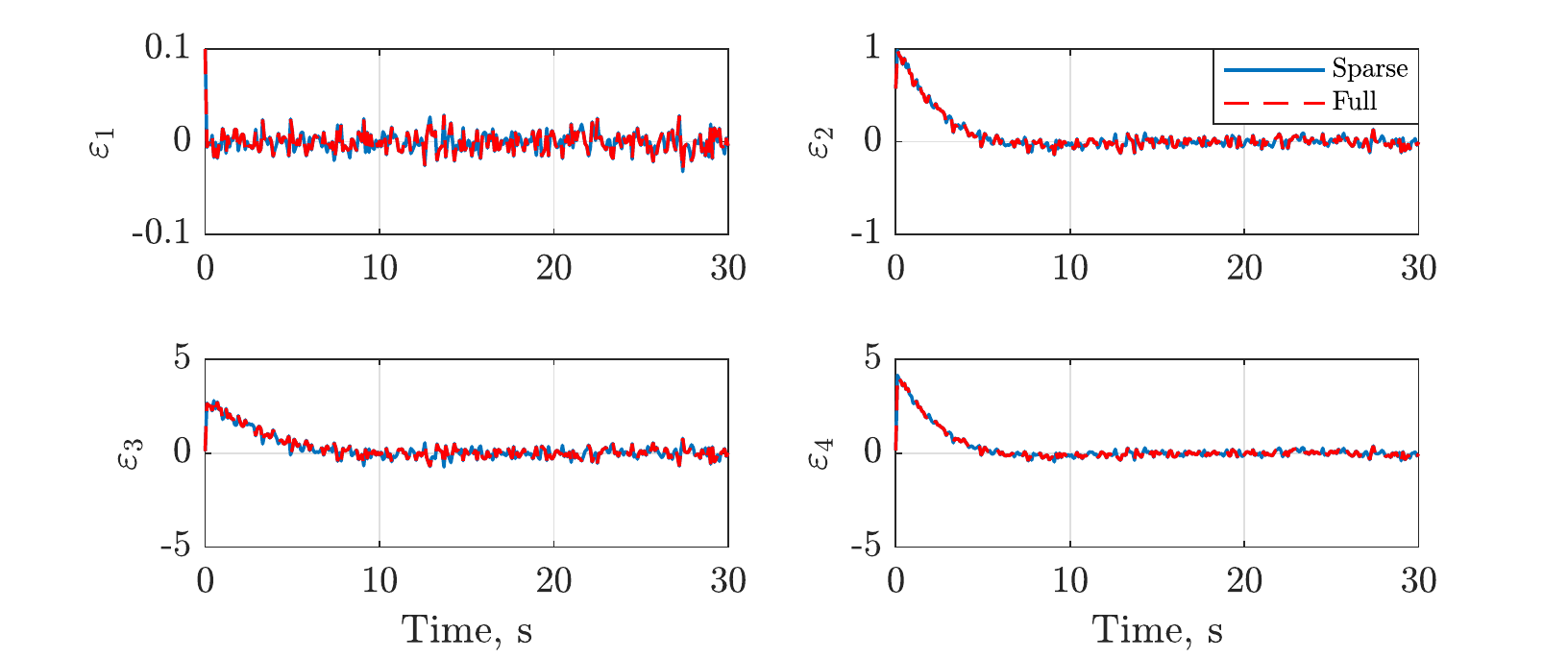}
    \caption{$\Hinf$ observer performance for sparse and full sensor configurations.} \label{fig:hinf_obs_compare}
\end{figure}

\subsubsection*{Augmented cost}
Next, we consider the augmented cost function defined in \eqn{aug_cost}, with upper bound on $\vo{\kappa}^2$. Let us assume that the upper bound is
\begin{align}
    \vo{\kappa}^2\leq \vo{\kappa}_{\text{max}}^2 = \begin{bmatrix}1 & 1 & 0.01 & 0.01 & 2.25\end{bmatrix}^T, \nonumber
\end{align}
and the inequality is elementwise. Such a constraint on $\vo{\kappa}^2$ may arise due to real world limitations, e.g. maximum possible precision with which a sensor can be manufactured. Therefore, \responseOne{the constraint} in terms of $\vo{\beta}$, which is the optimization parameter, can be written as
\begin{align*}
    \vo{\beta}-\vo{\kappa}_{\text{max}} ^2 &\leq 0\quad \text{ for $\Htwo$, and}, \\
    \vo{\beta}-\gamma\vo{\kappa}_{\text{max}} ^2 &\leq 0\quad \text{ for $\Hinf$}.
\end{align*}
The  cost function \eqn{aug_cost} is minimized for the sparse configuration identified in \fig{polished_h2_hinf} with bounds on $\vo{\beta}$ as defined above.
The optimal $\vo{\kappa}^2$ obtained for different values of $c$, for $\Htwo$ and $\Hinf$ observer design are shown in \fig{h2ex2} and \fig{hinfex2} respectively. The titles of subplots also show the corresponding optimal $\gamma$ obtained for the specified $c$.

From \fig{h2ex2}, it is clear that, as more \responseOne{weight} is given to minimizing $\gamma$ in the optimization problem, the optimal value of $\gamma$ decreases (i.e. performance of the observer improves), and the values of $\kappa^2$ increase (i.e. better performance requires higher precision).  We also observe that in the last plot corresponding to $c=1000$, the precision bounds are saturated, which means that $\gamma = 0.29$ is the best performance that can be achieved with this sensor configuration and the upper bounds $\vo{\kappa}_{\text{max}} ^2$. If better performance is desired, one should allow higher sensor precision or a sensor configuration with more number of sensors. This exposes a trade-off between the performance $\gamma$, and the minimal precision and sparse configuration quantified by the $l_1$-norm.
Similar observations can be made for \fig{hinfex2} as well.

\begin{figure}[htb]
    \centering
    \includegraphics[trim=2cm 0cm 3cm 0.2cm,width=0.45\textwidth]{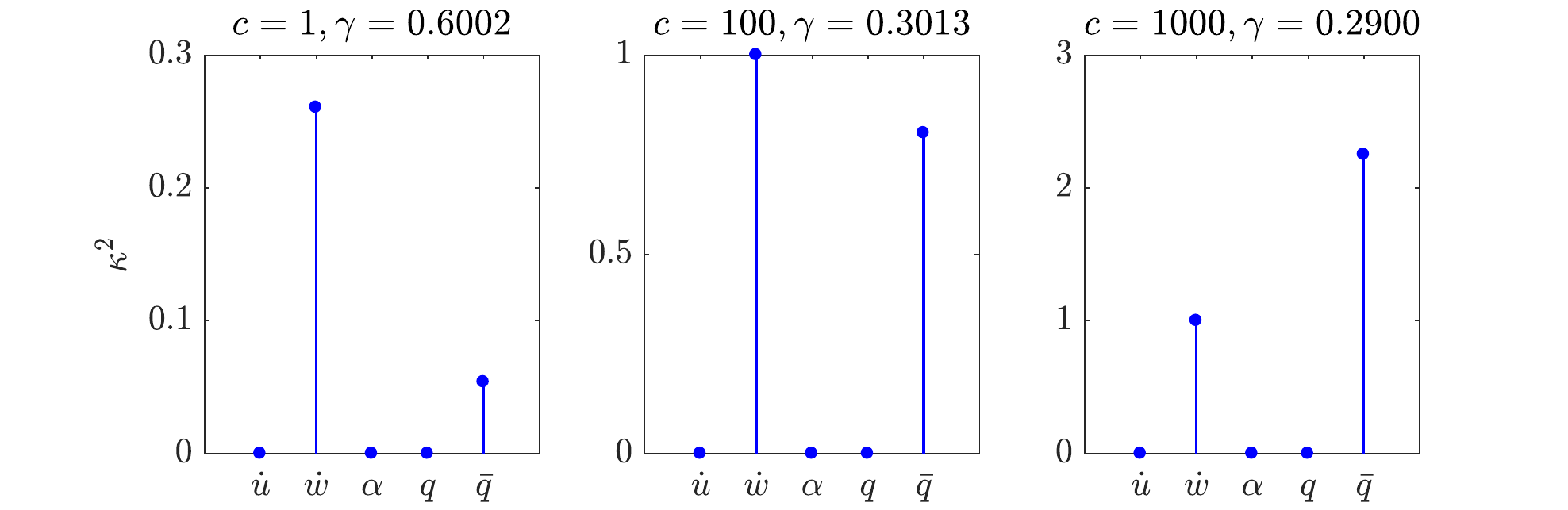}
    \caption{$\Htwo$ optimal precision, $\vo{\kappa}^2\leq \vo{\kappa}_{\text{max}}^2$.}
    \label{fig:h2ex2}
\vspace{1cm}
    \includegraphics[trim=2cm 0cm 3cm 0.2cm,width=0.45\textwidth]{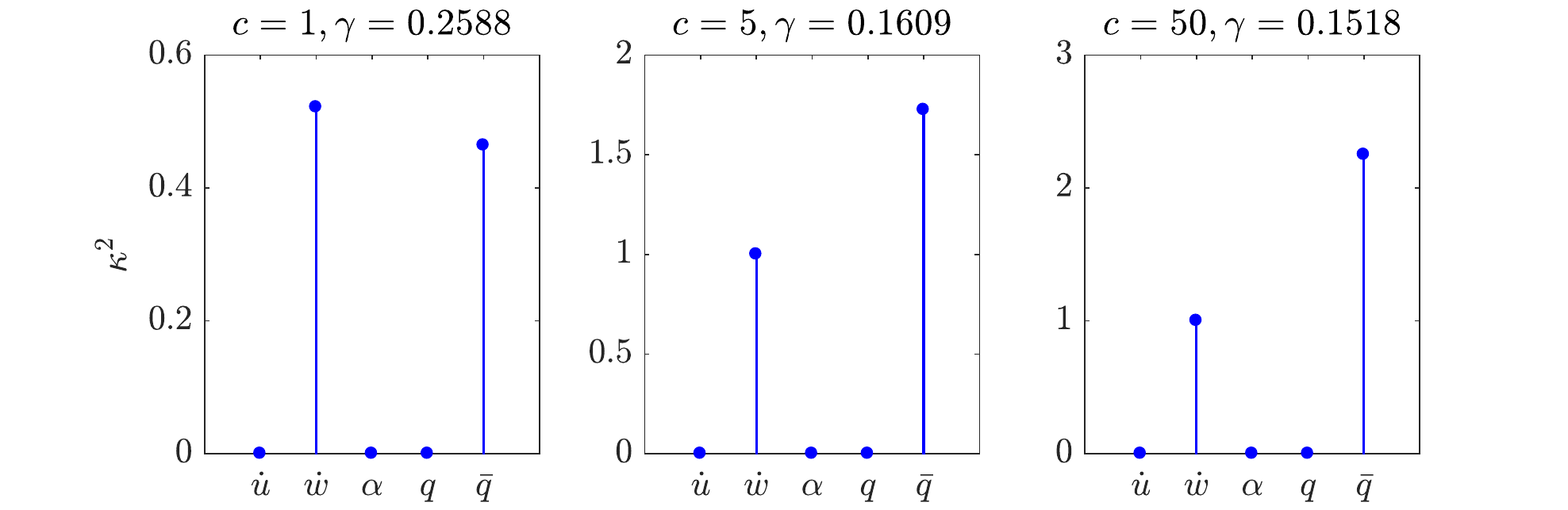}
    \caption{$\Hinf$ optimal precision, $\vo{\kappa}^2\leq \vo{\kappa}_{\text{max}}^2$.}
    \label{fig:hinfex2}
\end{figure}

\section{Conclusion} \label{sec:concl}
This paper presents an integrated theoretical framework to design $\Htwo/\Hinf$ optimal observers with sparse sensor configurations, while simultaneously minimizing the required sensor precision. The precision of sensor is treated as an optimization variable. A convex optimization problem is posed to minimize the sparsity-promoting $l_1$-norm of the sensor precision vector subject to linear matrix inequalities, and the sparse solution is obtained iteratively. Application of the proposed approach is demonstrated on a linearized model of an F-16 aircraft.
We also showed that the upper bounds on precision of sensors can be easily incorporated in the optimization problem, and the minimum possible attenuation level $\gamma$ can also be determined by augmenting the cost function.
For brevity of discussion, the development of customized algorithms to solve the optimization problem efficiently for large-scale systems was not discussed in this paper, and will be a topic of our future work.





\section*{APPENDIX} \label{appendix}
\noindent System matrices for the linearized F-16 model:
\begin{align*}
   { \A = \begin{bmatrix} -1.8969e{-02} & -0.40518 & -32.17  & 0.89146 \\
  -6.4397e{-05} & -1.61760 & 0 &  0.93254\\
            0  &          0 &           0 &  1\\
  0 & -2.36830 &           0 & -1.9696 \end{bmatrix}} \end{align*}

\begin{align*}
\Bu =  \begin{bmatrix}
   1.5700e{-03} &  4.7404e{-09} & 0 & 0 \\
   6.6374e{-01} & -3.1441e{-03} & 0 & -5.3433{e-01} \end{bmatrix}^T
\end{align*}

$\Bd = \Bu \left[ 0 \;\, 1 \right]^T$, $\Bv = \vo{0}\in \Real^{4\times 5}$, $\Bw = \left[\Bd\;\;\Bv\right]$

\begin{align*}
\Du =  \begin{bmatrix}
   1.5700e{-03} &  0 & 0 & 0 & 0\\
    6.5425e{-01} &  -3.1461 & 0 & 0 & 0 \end{bmatrix}^T
\end{align*}

$\Dd = \Du \left[ 0 \;\, 1 \right]^T$, $\Dv = \I{5}$, $\Dw = \left[\Dd\quad \Dv\right]$, $\Cz = \I{4}$

\begin{align*}
    & \Cy = \begin{bmatrix}   -0.019164 & -5.2893 & -32.17 &  3.7071\\
  -0.064340 & -1.6176e{+03} &  0.09713 &  932.5332\\
            0 &  1 &           0 &           0\\
            0 &           0 &           0 &  1\\
   1.7578 &           0 &           0 &           0 \end{bmatrix} \end{align*}

$\W_u = \diag(500,5)$, $\W_w = \diag(0.5, \I{5})$, \\

$\W_n = \I{5}$, $\W_d = 0.5$, $\W_w = \diag(\W_d, \W_n)$, \\

$\W_z = \diag(1/100, 180/5\pi, 180/5\pi, 180/2\pi)$.



\bibliographystyle{unsrt}
\bibliography{root}
\end{document}

%% file: macros.tex
\newcommand{\vo}[1]{\boldsymbol{#1}}
\newcommand{\mo}[1]{\boldsymbol{#1}}
\newcommand{\A}{\vo{A}}
\newcommand{\B}{\vo{B}}
\newcommand{\F}{\vo{F}}
\newcommand{\x}{\vo{x}}
\newcommand{\y}{\vo{y}}
\newcommand{\z}{\vo{z}}
\newcommand{\w}{\vo{w}}
\newcommand{\n}{\vo{n}}
\newcommand{\e}{\vo{e}}
\newcommand{\Lg}{\vo{L}} 

\newcommand{\nuu}{N_u}
\newcommand{\nv}{\ny}
\newcommand{\nx}{N_x}
\newcommand{\ny}{N_y}
\newcommand{\nz}{N_z}
\newcommand{\nw}{N_w}
\newcommand{\nd}{N_d}

\newcommand{\Htwo}{\mathcal{H}_2}
\newcommand{\Hinf}{\mathcal{H}_{\infty}}
\newcommand{\norm}[2]{\left\lVert#1\right\rVert_{#2}}

\newcommand{\Bd}{\vo{B}_d}
\newcommand{\Bu}{\vo{B}_u}
\newcommand{\Bv}{\vo{B}_n}
\newcommand{\Bw}{\vo{B}_w}
\newcommand{\Bwb}{\vo{B}_{\overline{\w}}}
\newcommand{\Cy}{\vo{C}_y}
\newcommand{\Cz}{\vo{C}_z}
\newcommand{\Dd}{\vo{D}_d}
\newcommand{\Du}{\vo{D}_u}
\newcommand{\Dv}{\vo{D}_n}
\newcommand{\Dw}{\vo{D}_w}
\newcommand{\Dwb}{\vo{D}_{\overline{\w}}}
\newcommand{\Sd}{\vo{S}_d}
\newcommand{\Sv}{\vo{S}_n}
\newcommand{\Sw}{\vo{S}_w}

\newcommand{\zerr}{\vo{\varepsilon}}
\newcommand{\xerr}{\e}
\newcommand{\Gwz}{\G_{\bar{\w}\rightarrow\zerr}(s)}
\newtheorem{theorem}{Theorem}

\newcommand{\xdot}{\dot{\vo{x}}}

\newcommand{\comment}[1]{\textcolor{red}{#1}}
\newcommand{\blue}[1]{\textcolor{blue}{#1}}

\newcommand{\response}[1]{{#1}}
\newcommand{\responseOne}[1]{{#1}}
\newcommand{\responseTwo}[1]{{#1}}
\newcommand{\responseThree}[1]{{#1}}
\newcommand{\responseFour}[1]{{#1}}

\newcommand{\set}[1]{\mathcal{#1}}
\newcommand{\Exp}[1]{\mathbb{E}\left[#1\right]}
\newcommand{\E}[1]{\Exp{#1}}

\newcommand{\I}[1]{\vo{I}_{#1}}
\newcommand{\X}{\vo{X}}

\newcommand{\W}{\vo{W}}
\newcommand{\Y}{\vo{Y}}
\newcommand{\R}{\vo{R}}
\newcommand{\M}{\vo{M}}
\newcommand{\Q}{\vo{Q}}

\renewcommand{\u}{\vo{u}}

\renewcommand{\P}{\mo{P}} 
\newcommand{\Real}{\mathbb R}

\newcommand{\inner}[1]{\left\langle \vo{e}\phi_i\right\rangle}

\newcommand{\eqnlabel}[1]{\label{eqn:#1}}

\newcommand{\eqn}[1]{(\ref{eqn:#1})}

\newcommand{\fig}[1]{Fig. (\ref{fig:#1})}

\DeclareMathAlphabet{\mathbfsf}{\encodingdefault}{\sfdefault}{bx}{n}

\newcommand{\C}{\vo{C}}
\newcommand{\D}{\vo{D}}
\newcommand{\G}{\vo{G}}

\newcommand{\domain}[1]{\set{D}}

\newcommand{\diag}{\textbf{diag}}

\newcommand{\trace}[1]{\textbf{trace}\left( #1 \right)}